%% file: main.tex
\documentclass[a4paper,onecolumn,unpublished,11pt]{quantumarticle}
\pdfoutput=1
\usepackage[english]{babel}
\usepackage[numbers,sort&compress]{natbib}
\usepackage{amsmath}
\usepackage{amssymb}
\usepackage{amsthm}
\usepackage{graphicx}
\usepackage{xcolor}
\usepackage{mathrsfs}
\usepackage{enumitem}
\usepackage{csquotes}
\usepackage{braket}
\usepackage{mathtools}
\usepackage{booktabs}
\usepackage{thmtools}
\usepackage{enumitem}
\usepackage{algorithm}

\usepackage{algpseudocode}
\usepackage{verbatim}
\usepackage{soul}
\usepackage{apptools}
\usepackage[most]{tcolorbox}
\usepackage{xcolor}
\usepackage[bookmarks=true,colorlinks=true,linkcolor=teal,citecolor=teal,urlcolor=teal,bookmarksopen=true,bookmarksopenlevel=1]{hyperref}
\usepackage[capitalize]{cleveref}

\input{commands.tex}

\newtheorem{theorem}{Theorem}

\newtheorem{definition}[theorem]{Definition}

\newtheorem{remark}[theorem]{Remark}

\newcommand{\maxlin}{\textup{max-LINSAT}} 
\newcommand{\maxxor}{\textup{max-XORSAT}} 

\newcommand{\NP}{\mathsf{NP}}
\newcommand{\Pclass}{\mathsf{P}}
\newcommand{\PneqNP}{\Pclass \neq \NP}
\newcommand{\OPT}{\mathrm{OPT}}

\definecolor{max}{RGB}{164,12,52}

\usepackage[most]{tcolorbox}

\tcbset {
  base/.style={
    arc=0mm, 
    bottomtitle=0.5mm,
    boxrule=0mm,
    colbacktitle=black!10!white, 
    coltitle=black, 
    fonttitle=\bfseries, 
    left=2.5mm,
    leftrule=1mm,
    right=3.5mm,
    title={#1},
    toptitle=0.75mm, 
  }
}

\definecolor{brandblue}{rgb}{0.34, 0.7, 1}
\newtcolorbox{mainbox}[1]{
  colframe=brandblue, 
  base={#1}
}

\newtcolorbox{subbox}[1]{
  colframe=black!30!white,
  base={#1}
}

\newcommand{\dccqs}{Dahlem Center for Complex Quantum Systems, Freie Universit{\"a}t Berlin, 14195 Berlin, Germany}
\newcommand{\hzb}{Helmholtz-Zentrum Berlin f{\"u}r Materialien und Energie, 14109 Berlin, Germany}
\newcommand{\hhi}{Fraunhofer Heinrich Hertz Institute, 10587 Berlin, Germany}
\newcommand{\tub}{Electrical Engineering and Computer Science, Technische Universität Berlin, 10587 Berlin, Germany}

\begin{document}

\title{Tight inapproximability of max-LINSAT\newline and implications for decoded quantum interferometry}

    \author{Maximilian\ J.\ Kramer}
    \email{m.kramer@fu-berlin.de}
    \affiliation{\dccqs}

    \author{Carsten Schubert}
    \affiliation{\tub}

    \author{Jens\ Eisert}
    \affiliation{\dccqs}
    \affiliation{\hzb}
    \affiliation{\hhi}

\begin{abstract}
We establish tight inapproximability bounds for max-LINSAT, the problem of maximizing the number of satisfied linear constraints over the finite field $\mathbb{F}_q$, where each constraint accepts $r$ values.
Specifically, we prove by a direct reduction from H\r{a}stad's theorem that no polynomial-time algorithm can exceed the random-assignment ratio $r/q$ by any constant, assuming $\mathsf{P} \neq \mathsf{NP}$. 
This threshold coincides with the $\ell/m \to 0$ limit of the semicircle law governing decoded quantum interferometry (DQI), where $\ell$ is the decoding radius of the underlying code.
Together, these observations delineate the boundary between worst-case hardness and potential quantum advantage, showing that any algorithm surpassing $r/q$ must exploit instance structure beyond what is present in the hard instances produced by PCP reductions.
\end{abstract}

\maketitle

\section{Introduction}

Recent years have seen growing interest in quantum algorithms for combinatorial optimization~\cite{Abbas_2024}, particularly for problems with strong worst-case classical hardness. 
Given the practical and industrial relevance of these problems across a wide range of scheduling and routing tasks, the design of corresponding quantum algorithms is of considerable importance. For such problems, one cannot expect a quantum computer to solve or approximate all instances of the combinatorial optimization problem in polynomial time. This does not mean that one cannot hope that quantum computers will have computational advantages over classical computers. Indeed, in such cases, potential super-polynomial quantum advantages are often investigated on restricted or structured families of instances~\cite{Pirnay2024, Yamakawa2024, szegedy2022, buhrman2025formalframeworkquantumadvantage}. One such example that has recently been given a lot of attention is the \emph{decoded quantum interferometry} (DQI) algorithm introduced in Ref.~\cite{Jordan2024DQI}. 
DQI works by using the quantum Fourier transform to reduce the optimization problem to a decoding task, and then coherently applies an efficient decoding algorithm to amplify the quantum interference between satisfying assignments. The approximation ratio of DQI is characterized by a closed-form formula, coined the \emph{semicircle law}~\cite{Jordan2024DQI}.

DQI targets a generalized linear constraint satisfaction problem over finite fields, commonly referred to as \maxlin. The flagship application is the \emph{optimal polynomial intersection} (OPI) {problem}, a structured \maxlin{} instance whose constraint matrix has Vandermonde structure arising from Reed--Solomon codes that are efficiently decodable up to their designed error radius.
For OPI, in a certain parameter regime, DQI achieves a super-polynomial speed-up over \emph{known} classical algorithms.

A growing body of follow-up work has developed around DQI~\cite{chailloux2024softdecoders, Patamawisut2025, ralli2025DQI, bu2025DQInoise, briaud2025quantumadvantagesolvingmultivariate, sabater2025solvingindustrialintegerlinear, marwaha2025complexitydecodedquantuminterferometry, anschuetz2025DQI, piveteau2025_quantum_decoding, parekh2025DQI_maxcut, gu2025algebraicgeometrycodesdecoded, kothari2025exponentialquantumspeedupmathrmsisinfty, schmidhuber2025hamiltoniandecodedquantuminterferometry, khattar2025verifiablequantumadvantageoptimized, chailloux2025opixsoftdecoders, rosmanis2026nearlylineartimedecodedquantum, bu2026hamiltoniandecodedquantuminterferometry}.
Most relevant to our setting, the DQI advantage for OPI has been strengthened via soft decoding of Reed--Solomon codes~\cite{chailloux2024softdecoders, chailloux2025opixsoftdecoders} and extended to algebraic geometry codes via a generalization called HOPI~\cite{gu2025algebraicgeometrycodesdecoded}. At the same time, several works have exposed DQI's limitations \cite{parekh2025DQI_maxcut, anschuetz2025DQI, marwaha2025complexitydecodedquantuminterferometry}.

A prerequisite for precisely interpreting the results surrounding DQI is a clear understanding of the \emph{approximability limits} of the underlying optimization problem. For many well-studied constraint satisfaction problems, tight inapproximability thresholds are known, often matching the performance of a random assignment. For \maxlin{} in the formulation considered by DQI, which allows linear constraints with arbitrary acceptance sets over a finite field, no such worst-case guarantees have previously been established. Without such worst-case bounds, it is unclear whether the approximation ratios achieved by quantum or classical heuristics are close to being optimal, or whether substantially better polynomial-time algorithms might exist.

We address this gap in the present work by establishing \emph{tight inapproximability results} for \maxlin{} over finite fields. We consider \maxlin{} instances over a finite field\footnote{Recall that the finite field $\mathbb{F}_q$ of order $q$ exists whenever $q = p^\tau$ for some prime $p$ and integer $\tau\geq 1$.} $\mathbb{F}_q$ in which each constraint is satisfied whenever a linear form evaluates to one of $r$ allowed values, with $1 \le r \le q-1$. This class includes, as a special case, the linear constraint satisfaction problems studied in Ref.~\cite{Jordan2024DQI}. A random assignment satisfies each constraint with probability exactly $r/q$, yielding a natural baseline approximation ratio.

Our main result shows that this random baseline is \emph{provably optimal in the worst case}. Specifically, for every finite field $\mathbb{F}_q$ of order $q$, every $1 \le r \le q-1$, and every $\varepsilon > 0$, it is $\NP$-hard to distinguish between instances in which almost all constraints are satisfiable and instances in which no assignment satisfies more than a $(r/q+\varepsilon)$-fraction of constraints. Assuming $\Pclass \neq \NP$, no polynomial-time classical algorithm can approximate \maxlin{} consistently beyond the random-assignment threshold, and likewise no polynomial-time quantum algorithm can do so unless $\NP \subseteq \mathsf{BQP}$. From the perspective of quantum optimization, this provides a \emph{complexity-theoretic benchmark}. The threshold $r/q$ is hence a hard wall for worst-case \maxlin$(q,r)$: no efficient algorithm can consistently exceed it on arbitrary instances. 
This does not conflict with DQI's demonstrated advantage on OPI, which operates in a qualitatively different regime with constant $n/m$, and instances derived from classical Reed--Solomon codes that exhibit specific Vandermonde structure. For such instances, both DQI and classical algorithms achieve approximation ratios exceeding the random assignment baseline, whereas the hard instances produced by PCP-reductions have $n/m = o(1)$.
Rather, our result delineates a boundary: while structured instances may admit algorithms that beat the random-assignment threshold, no polynomial-time algorithm can do so on arbitrary \maxlin{} instances under reasonable complexity-theoretic assumptions.

Our proof builds on and generalizes an inapproximability result of H{\r{a}}stad for \textup{max-E}$3$-\textup{LIN} over finite Abelian groups, which establishes tight hardness for the case of single-valued acceptance sets~\cite{hastad2001}. We show how this hardness extends to \maxlin{} instances with uniform acceptance sets of size $r$, yielding a matching hardness threshold for the general problem. The resulting inapproximability bounds are tight.

\section{Preliminaries}
We start by defining the optimization problems considered in this work, in particular focusing on the \maxlin-problem as defined in Ref.~\cite{Jordan2024DQI}.
\begin{definition}[\maxlin{} over $\mathbb{F}_q$]
    Let $\mathbb{F}_q$ be the finite field of order $q$ and let $B \in \mathbb{F}_q^{m \times n}$. For each $i=1,\ldots,m$, let $F_i \subset \mathbb{F}_q$ be an arbitrary subset of $\mathbb{F}_q$, which yields a corresponding constraint $\sum_{j=1}^n B_{i, j} x_j \in F_i$. The \maxlin{} task is to find $\xx \in \mathbb{F}_q^n$ satisfying as many as possible of these $m$ constraints. We denote by $\OPT$ the maximum number of satisfiable constraints and define an approximation ratio $\alpha$ for a solution satisfying $s$ out of $m$ constraints to be $\alpha \coloneq s / m $.
\end{definition}
We assume that the acceptance sets $F_i$ are given as an explicit list.
\begin{definition}[\maxlin$(q,r)$]\label{def:max-linsat-qr}
We write \maxlin$(q,r)$ for the 
restriction of \maxlin{} over $\mathbb{F}_q$ to instances in which $|F_i|=r$ for every~$i$. The special case \maxxor~is \maxlin$(2,1)$: the field is $\mathbb{F}_2$ and every acceptance 
set is a singleton, so each constraint is a linear equation over~$\mathbb{F}_2$.
\end{definition}
We will relate this problem to \textup{max-E}$k$\textup{-LIN-}$q$ for some integer $k$ which has been studied more extensively in the literature.

\begin{definition}[\textup{max-E}$k$\textup{-LIN-}$q$]\label{def:ek-lin-q}
Let $\mathbb{F}_q$ be the finite field of order $q$ and let $k \ge 1$ be an integer. An instance of \textup{max-E}$k$\textup{-LIN-}$q$ consists of $n$ variables $x_1,\dots,x_n \in \mathbb{F}_q$ and $m$ linear equations
\begin{align}
\sum_{j \in S_i} x_j &= b_i,
\qquad i = 1,\dots,m,
\end{align}
where for each $i$, $S_i \subseteq \{1,\dots,n\}$ satisfies $|S_i| = k$, and $b_i \in \mathbb{F}_q$.
The goal is to find an assignment $\xx \in \mathbb{F}_q^n$ that satisfies the maximum number of equations.
\end{definition}
To establish worst-case limits on the approximation for \maxlin, we rely on a foundational result of H\r{a}stad~\cite[Theorem 5.9]{hastad2001}, which proves strong inapproximability for \textup{max-E}$3$-LIN-$\Gamma$ over finite Abelian groups. H\r{a}stad's result builds on the celebrated PCP theorem~\cite{AS98,ALMSS98}, which characterizes $\NP$ in terms of probabilistically checkable proofs and underlies most modern inapproximability results. This result provides the base case from which we extend hardness to \maxlin{} instances with uniform acceptance sets of arbitrary size.

\begin{theorem}[Inapproximability of \textup{max-E}$3$-LIN-$\Gamma$]
\label{thm:hastad}
    For every finite Abelian group~$\Gamma$ and every $\varepsilon>0$, it is $\NP$-hard to approximate \textup{max-E}$3$\textup{-LIN-}$\Gamma$ within a factor $|\Gamma|-\varepsilon$. Equivalently, \textup{max-E}$3$\textup{-LIN-}$\Gamma$ is non-approximable beyond the random assignment threshold: it is $\NP$-hard to distinguish, given an instance of \textup{max-E}$3$\textup{-LIN-}$\Gamma$, between
    \begin{itemize}
       \item[\textup{(Y)}] instances with $\OPT \ge (1-\varepsilon)\,m$ and
       \item[\textup{(N)}] instances with $\OPT \le (1/|\Gamma|+\varepsilon)\,m$.
    \end{itemize}
    
    In particular, setting $\Gamma=(\mathbb{F}_q, +)$ for a finite field $\mathbb{F}_q$ yields the $(1-\varepsilon,\;1/q+\varepsilon)$-hardness of \textup{max-E}$3$\textup{-LIN-}$q$.
\end{theorem}
The inapproximability guarantee in \cref{thm:hastad} is \emph{tight}. Indeed, for \textup{max-E}$3$\textup{-LIN-}$\Gamma$ over a finite Abelian group $\Gamma$, a random assignment satisfies each equation with probability exactly $1/|\Gamma|$. Therefore, any polynomial-time algorithm can trivially achieve an approximation ratio of $1/|\Gamma|$. H\r{a}stad's theorem shows that it is $\NP$-hard to exceed this baseline by more than an arbitrarily small constant $\varepsilon$, so the result matches the random-assignment threshold up to $\varepsilon$. No stronger worst-case hardness can hold under standard complexity assumptions, making the theorem optimal in this sense.

The inapproximability of \maxlin$(q,r)$ can also be situated within the broader theory of constraint satisfaction problems.
Under the \emph{unique games conjecture} (UGC)~\cite{Khot2002UGC}, Ref.~\cite{Raghavendra2005} showed that for every \emph{constraint satisfaction problem} (CSP), the optimal poly\-nomial-time approximation ratio equals the integrality gap of a canonical \emph{semi-definite programming} (SDP) relaxation; moreover, this ratio is achieved by an appropriate rounding of the SDP solution.
More concretely, Ref.~\cite{Austrin2008} has proven that any predicate whose satisfying assignments support a balanced pairwise-independent distribution is approximation resistant under the UGC. 
The \maxlin$(q,r)$ predicate for arity $k=3$ and coefficients being one, i.e., checking whether $x_{i_1}+x_{i_2}+x_{i_3}$ falls in an $r$-element subset $S \subseteq \mathbb{F}_q$, satisfies this criterion (as does the analogous predicate for any fixed $k \geq 3$), so UGC-hardness of beating $r/q$ follows immediately from their result.
Our \cref{thm:inapproximability_of_max_linsat} establishes the same conclusion under the weaker assumption $\PneqNP$, via a direct reduction from H\r{a}stad's theorem~\cite{hastad2001}. This is notable because the UGC remains disputed.
In a different direction, Ref.~\cite{Chan2016} proved unconditional ($\PneqNP$) approximation resistance for predicates whose satisfying assignments contain a pairwise-independent subgroup.
For prime~$q$, however, the only additive subgroups of~$\mathbb{F}_q$ are $\{0\}$ and $\mathbb{F}_q$ itself, so acceptance sets of size $1 \leq r \leq q-1$ are never subgroups, and this criterion does not apply to our setting. 
For prime power $q$, nontrivial additive subgroups do exist, but the criterion still does not cover generic acceptance sets of arbitrary size $r$.
Ref.~\cite{Engebretsen2004} extended H\r{a}stad's \textup{max-E}$3$-\textup{LIN} hardness from Abelian to all finite groups, and Ref.~\cite{Butti2025} recently transferred these inapproximability results to the setting of \emph{promise constraint satisfaction problems}. 
These results may prove useful for establishing inapproximability 
bounds for potential DQI variants over 
more general finite groups.

\section{Results}
\begin{theorem}[Inapproximability of \maxlin$(q,r)$]
\label{thm:inapproximability_of_max_linsat}
For every finite field $\mathbb{F}_q$, every integer $1\leq r \leq q-1$, and every $\varepsilon>0$, it is $\NP$-hard to distinguish, given an instance of \maxlin$(q,r)$, between
\begin{itemize}
\item[\textup{(Y)}] instances with $\OPT \ge (1-\varepsilon)\,m$ and
\item[\textup{(N)}] instances with $\OPT \le (r/q+\varepsilon)\,m$.
\end{itemize}
This means it is $\NP$-hard to approximate within a factor strictly better than $r/q+\varepsilon$ for every $\varepsilon>0$.
\end{theorem}

\begin{proof}
The case $r=1$ follows directly from \cref{thm:hastad}, since every instance of \textup{max-E}$3$\textup{-LIN-}$q$ is syntactically an instance of \maxlin$(q,1)$ with coefficients $B_{i,j}$ being one if and only if the variable is involved: each equation $\sum_{j \in S_i} x_j = b_i$ defines a constraint with $F_i = \{b_i\}$.

For general $r \geq 2$, we give a deterministic polynomial-time reduction from \maxlin$(q,1)$ to \maxlin$(q,r)$. Let $\phi$ be an instance of \maxlin$(q,1)$ with $m$ constraints $L_i(\xx) = b_i$ for $i = 1, \dots, m$, where each $L_i$ is a linear form over $\mathbb{F}_q$. We construct an instance $\phi'$ of \maxlin$(q,r)$ on the same set of $n$ variables as follows. For each constraint $i$, we enumerate all $r$-element subsets of $\mathbb{F}_q$ that contain~$b_i$. There are exactly $\binom{q-1}{r-1}$ such subsets, since we choose the remaining $r-1$ elements from $\mathbb{F}_q \setminus \{b_i\}$. For each such subset $S$, we add the constraint $L_i(\xx) \in S$ to $\phi'$. The resulting instance $\phi'$ has $m' = m \cdot \binom{q-1}{r-1}$ constraints, each with acceptance set of size exactly~$r$. Since $q$ is a fixed constant, $m'$ is polynomial in $m$.

\emph{Completeness.}
If $\xx$ satisfies $L_i(\xx) = b_i$, then $L_i(\xx) \in S$ for every $r$-element subset $S$ containing $b_i$, so all $\binom{q-1}{r-1}$ derived constraints from constraint $i$ are satisfied. Therefore, $\OPT(\phi') \geq \OPT(\phi) \cdot \binom{q-1}{r-1}$. In the (Y)-case, $\OPT(\phi) \geq (1-\varepsilon)\,m$, so that
\begin{equation}
    \OPT(\phi') \geq (1-\varepsilon)\,m'.
\end{equation}
\emph{Soundness.}
Fix any assignment $\xx \in \mathbb{F}_q^n$ and define
\begin{align}
    A(\xx) &= \{ i : L_i(\xx) = b_i \}, \\
    \overline{A}(\xx) &= \{ i : L_i(\xx) \neq b_i \}.
\end{align}
For each $i \in A(\xx)$, all $\binom{q-1}{r-1}$ derived constraints are satisfied. For each $i \in \overline{A}(\xx)$, let $c_i = L_i(\xx) \neq b_i$. The derived constraint $L_i(\xx) \in S$ is satisfied if and only if $c_i \in S$. Since $S$ must contain $b_i$, the number of $r$-element subsets of $\mathbb{F}_q$ containing both $b_i$ and $c_i$ is $\binom{q-2}{r-2}$ (choosing the remaining $r-2$ elements from $\mathbb{F}_q \setminus \{b_i, c_i\}$). Thus, for each $i \in \overline{A}(\xx)$, exactly $\binom{q-2}{r-2}$ out of $\binom{q-1}{r-1}$ derived constraints are satisfied. Consequently, the fraction of derived constraints satisfied per unsatisfied original constraint is
\begin{align}
    \frac{\binom{q-2}{r-2}}{\binom{q-1}{r-1}} = \frac{r-1}{q-1}.
\end{align}
Writing $\mu \coloneqq |A(\xx)|/m$ for the fraction of original constraints satisfied exactly, the total fraction of constraints in $\phi'$ satisfied by $\xx$ is
\begin{align}
    \mu \cdot 1 + (1-\mu) \cdot \frac{r-1}{q-1}
    = \mu \cdot \frac{q-r}{q-1} + \frac{r-1}{q-1}.
\end{align}
In the (N)-case, \cref{thm:hastad} guarantees $\mu \leq 1/q + \varepsilon$ for every assignment. By substituting this and using the identity
\begin{align}
    \frac{1}{q} \cdot \frac{q-r}{q-1} + \frac{r-1}{q-1}
    = \frac{(q-r) + q(r-1)}{q(q-1)}
    = \frac{r}{q},
\end{align}
we obtain
\begin{align}
    \mu \cdot \frac{q-r}{q-1} + \frac{r-1}{q-1}
    \leq \left(\frac{1}{q} + \varepsilon\right) \frac{q-r}{q-1} + \frac{r-1}{q-1} = \frac{r}{q} + \varepsilon \cdot \frac{q-r}{q-1}.
\end{align}
Since $(q-r)/(q-1) \leq 1$, this is at most $r/q + \varepsilon$. Therefore, $\OPT(\phi') \leq (r/q + \varepsilon)\,m'$, completing the soundness analysis.
\end{proof}

\begin{remark}[Tightness of \cref{thm:inapproximability_of_max_linsat}]\label{cor:tight}
The threshold $r/q$ in \cref{thm:inapproximability_of_max_linsat} is tight. Uniformly chosen assignments $\xx \in\mathbb{F}_q^n$ satisfy each constraint~$i$ with probability exactly $|F_i|/q = r/q$, giving an expected approximation ratio of~$r/q$. A deterministic polynomial-time algorithm achieving a ratio of at least~$r/q$ is obtained by the method of conditional expectations.
\end{remark}

\begin{remark}[Coefficient-free hard instances]\label{rem:coefficient-free}
The hard instances produced by the reduction have a particularly simple form: the constraint matrix~$B$ has entries in~$\{0,1\}$ only, with exactly three ones per row. This follows from the fact that H\r{a}stad's construction yields coefficient-free equations $x_{i_1} + x_{i_2} + x_{i_3} = b_i$. The inapproximability of \maxlin$(q,r)$ therefore already holds for this restricted class of instances, which is a stronger statement than hardness for general constraint matrices.
\end{remark}

\begin{remark}[\maxxor]\label{rem:xorsat}
Setting $q=2$ and $r=1$ in \cref{thm:inapproximability_of_max_linsat} recovers the result from Ref.~\cite{hastad2001} stating that \maxxor – equivalently, \textup{max-E}$3$\textup{-LIN-}$2$ – is $\NP$-hard to approximate within $1/2+\varepsilon$ for every $\varepsilon>0$.
\end{remark}

\section{Discussion}

In this work, we have established tight inapproximability results for a family of optimization problems that feature strongly in the \emph{decoded quantum interferometry} (DQI) algorithm~\cite{Jordan2024DQI}, which has been proposed to offer a speedup over classical approaches. \Cref{thm:inapproximability_of_max_linsat} establishes that no polynomial-time algorithm can approximate the problem \maxlin$(q,r)$ beyond the random assignment threshold $r/q$ on worst-case instances, assuming $\Pclass \neq \NP$, and likewise no polynomial-time quantum algorithm can do so unless $\NP \subseteq \mathsf{BQP}$. In light of the fact that this kind of problem is central to DQI, we now discuss the implications of this result for evaluating quantum optimization algorithms by resorting to DQI and beyond.

A key observation is that DQI's quantum advantage claims are \emph{not} for arbitrary \maxlin{} instances, but for \emph{structured problems}, most notably \emph{optimal polynomial intersection} (OPI), with constant $n/m$ and a constraint matrix that has Vandermonde structure arising from classical Reed--Solomon codes, which DQI exploits via coherent implementation of the Berlekamp--Massey decoding algorithm~\cite{Jordan2024DQI}. For OPI with parameters $n/q = r/q = 1/2$, DQI achieves an approximation ratio of approximately $0.933$ via the \emph{semicircle law}, while the best known classical algorithm (Prange's information set decoding \cite{Prange}) achieves only ${\sim}0.55$~\cite{Jordan2024DQI}. These insights have recently been extended to algebraic geometry codes in Ref.~\cite{gu2025algebraicgeometrycodesdecoded}, where a problem called HOPI is studied, which generalizes OPI. Therefore, DQI targets structured instances, not worst-case \maxlin.

Our worst-case inapproximability result is entirely consistent with these claims: \cref{thm:inapproximability_of_max_linsat} applies to worst-case \maxlin$(q,r)$ instances, whereas OPI operates in the structured regime described above.
The hardness reduction from \textup{max-E}$3$\textup{-LIN-}$q$ produces instances with $n/m = o(1)$ and no such \emph{exploitable} structure.

In DQI, the optimization problem is mapped to a decoding task for a classical error-correcting code. A central parameter governing the algorithm's performance is the \emph{decoding radius}~$\ell$, defined as the maximum number of errors that an efficient decoder can correct. 
If a \maxlin{} instance with $m$ constraints gives rise to a dual code $C^\perp$ that can be decoded from up to $\ell$ errors, then DQI exploits this decodability to prepare a degree-$\ell$ polynomial enhancement of the objective function in superposition.
The resulting approximation ratio is given by the \emph{semicircle law}~\cite{Jordan2024DQI},
\begin{align}\label{eq:semicircle_discussion}
  \alpha_{\mathrm{DQI}}
  = \left(
      \sqrt{\frac{\ell}{m}\!\left(1 - \frac{r}{q}\right)}
    + \sqrt{\frac{r}{q}\!\left(1 - \frac{\ell}{m}\right)}
  \right)^{\!2}
\end{align}
if $r/q \leq 1 - \ell/m$, and $\alpha_{\mathrm{DQI}} = 1$ otherwise. This relationship is illustrated in \cref{fig:dqi_inapprox}.
The ratio $\ell/m$ therefore quantifies the amount of decodable structure available to the algorithm. When $\ell/m$ is bounded away from zero, the decoder can correct a constant fraction of errors, and DQI achieves approximation ratios strictly above the random-assignment baseline~$r/q$. In the limit $\ell/m \to 0$, the semicircle law yields $\alpha_{\mathrm{DQI}} \to r/q$: without sufficient decoding capability, DQI reduces to a random assignment.

\begin{figure}
    \centering
    \includegraphics[width=0.7\linewidth]{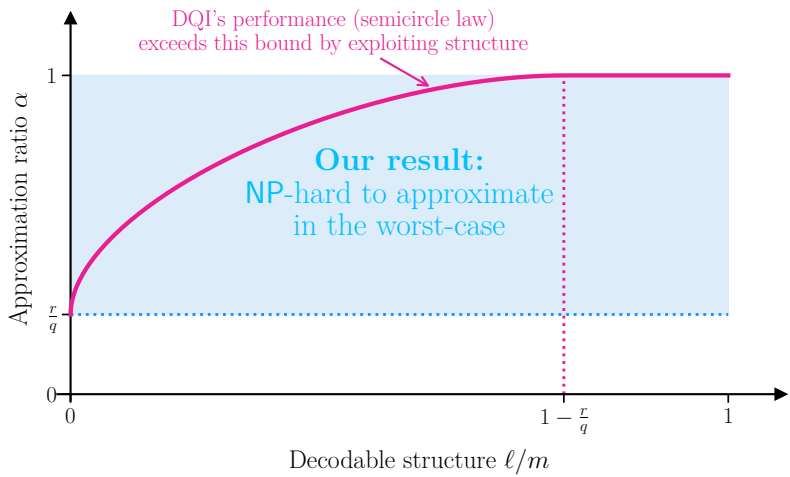}
    \caption{Landscape of approximability for $\mathrm{max\text{-}LINSAT}(q,r)$. The x-axis parametrizes the decodable structure $\ell/m$, where $\ell$ is the decoding radius of the underlying error-correcting code and $m$ is the number of constraints.
    For worst-case instances, by virtue of \cref{thm:inapproximability_of_max_linsat}, approximation ratios within the blue-shaded region are unattainable.  
    Importantly, this hardness statement does not depend on the parameter $\ell/m$: the x-axis is only meaningful for structured instances arising from decoding problems (such as those targeted by DQI). The shading therefore illustrates a worst-case approximation barrier rather than a property that varies with $\ell/m$.
    The pink curve depicts DQI's approximation ratio as given by the semicircle law~\cite{Jordan2024DQI}, which surpasses the random-assignment bound $r/q$ whenever $\ell/m > 0$ and saturates at $\alpha_{\mathrm{DQI}} = 1$ for $\ell/m \geq 1 - r/q$. In the limit $\ell/m \to 0$, DQI's performance degrades to exactly $r/q$, matching the worst-case barrier.}
    \label{fig:dqi_inapprox}
\end{figure}

We note that on random \maxlin$(q,r)$ instances with $n$ variables and $m$ constraints, simple algorithms can exceed the $r/q$ threshold. For example Prange's information set decoding~\cite{Prange} solves an $n \times n$ subsystem exactly and achieves a total expected approximation ratio of $n/m + (1 - n/m) \cdot r/q$ for random instances. This does not conflict with \cref{thm:inapproximability_of_max_linsat}: the formula is a heuristic expected-value analysis that assumes the residual $m - n$ constraints behave as random evaluations~\cite[Section~11.3]{Jordan2024DQI}, an assumption that fails on worst-case instances. Moreover, for any fixed $\varepsilon > 0$, the PCP-based reduction produces instances in which $n/m$ can be driven below~$\varepsilon$, so Prange's advantage vanishes in exactly the regime where the hardness applies.

This observation is complemented by \cref{thm:inapproximability_of_max_linsat}, which shows that on worst-case \maxlin$(q,r)$ instances, no polynomial-time algorithm can exceed the $r/q$ threshold (assuming $\mathrm{P} \neq \mathrm{NP}$). While these two statements operate in different regimes -- the semicircle law describes DQI's performance on code-derived instances as a function of the decoding radius, and \cref{thm:inapproximability_of_max_linsat} applies to adversarially constructed instances with no exploitable structure -- they point to the same conclusion from different directions: $r/q$ is the natural worst-case approximation barrier for \maxlin, and surpassing it requires both an algorithm capable of exploiting structure and instances that exhibit such structure. The codes for which DQI demonstrates an advantage over known classical algorithms, namely Reed--Solomon codes in OPI~\cite{Jordan2024DQI} and algebraic geometry codes in HOPI~\cite{gu2025algebraicgeometrycodesdecoded}, satisfy these conditions.

It is suggestive to read this connection in coding-theoretic terms. For code families where efficient decoding beyond a vanishing fraction of errors is believed to be computationally hard -- such as generic linear codes specified by arbitrary parity-check matrices, for which bounded-distance decoding is $\NP$-complete~\cite{berlekamp1978} -- one would expect the associated \maxlin{} instances to be similarly hard to approximate beyond~$r/q$, since a sufficiently good approximate solution would be expected to reveal information about nearby codewords. Making this intuition precise by exhibiting a formal reduction from approximating code-derived \maxlin{} instances to decoding remains an interesting open question.

Recent work has further clarified that DQI's advantage, where it exists, depends critically on problem structure:
Ref.~\cite{parekh2025DQI_maxcut} shows that for MaxCut instances where DQI achieves any nontrivial asymptotic approximation guarantee, the optimal solution can be found \emph{exactly} in classical polynomial time. DQI consequently provides no advantage for MaxCut.
For unstructured max-$k$-XOR-SAT instances, Ref.~\cite{anschuetz2025DQI} proves that DQI is obstructed by the overlap gap property, which is a topological barrier in the near-optimal solution space. Numerical evidence suggests that classical approximate message passing outperforms DQI on such instances.
From the perspective of computational complexity, Ref.~\cite{marwaha2025complexitydecodedquantuminterferometry} places DQI in $\mathsf{BPP}^{\mathsf{NP}}$, thereby ruling out immediate quantum advantage via ``supremacy''-style hardness 
arguments as they are used in quantum random sampling \cite{SupremacyReview,Boixo}. There, it is $\# \Pclass$-hard to compute individual output probabilities, and an efficient classical sampling algorithm up to a constant error in the total variation distance would lead to a collapse of the polynomial hierarchy to the third level. 
In contrast, DQI is already placed in a low level of the polynomial hierarchy. The hardness of DQI, if any, arises from locating elements in an exponentially large hidden subset without group structure and not from the hardness of simulating the quantum circuit itself, as here, in contrast to the situation in quantum random sampling, individual probabilities of outcomes can be efficiently computed.

Together with \cref{thm:inapproximability_of_max_linsat}, these results rigorously establish that DQI is not a general-purpose quantum optimizer. Rather, any possible quantum advantage over (known) classical optimizers crucially depends on structures that are present in the instances.
Several natural questions remain:
\begin{itemize}
    \item \textbf{Inapproximability for OPI and HOPI:}
    Can inapproximability results be extended to settings studied in OPI~\cite{Jordan2024DQI}? 
    The hard instances produced by our reduction have $n/m = o(1)$ and unstructured constraint matrices, so they do not directly address the Vandermonde-structured instances  with constant $n/m$ targeted by DQI. Since maximum-likelihood decoding of Reed--Solomon codes is $\NP$-hard for certain parameters~\cite{RS_nphard}, it is natural to ask whether PCP-based methods can be adapted to establish worst-case inapproximability for OPI. The same question extends to the HOPI generalization of Ref.~\cite{gu2025algebraicgeometrycodesdecoded}.
    
    \item \textbf{Average-case hardness:} Our result is worst-case. Establishing average-case inapproximability for natural distributions over \maxlin{} instances as they occur in DQI would provide stronger benchmarks for heuristic quantum algorithms. 

    \item \textbf{Heterogeneous acceptance sets:} Throughout this work (and also in Ref.~\cite{Jordan2024DQI}) the acceptance sets satisfy $|F_i|=r$ uniformly. If the cardinalities $|F_i|$ are allowed to vary across constraints, a random assignment satisfies constraint~$i$ with probability $|F_i|/q$, so the natural baseline is no longer a single reference value but depends on the distribution of the $|F_i|$. In particular, instances with small acceptance sets ($|F_i|=1$) are the most restrictive locally, yet mixtures of small and large sets may exhibit different global approximation behavior. What is a useful notion of a tight inapproximability threshold in this heterogeneous setting? Should it be governed by the minimum $|F_i|$, the average acceptance probability $\frac{1}{m}\sum_i |F_i|/q$, or a more refined structural parameter? Determining an optimal approximation guarantee as a function of the profile $(|F_i|)_{i=1}^m$ thus remains open.
\end{itemize}

Our \cref{thm:inapproximability_of_max_linsat} provides a complexity-theoretic benchmark for \maxlin, establishing $r/q$ as the provable worst-case approximation limit and confirming that any improvement must arise from instance structure not present in the hard instances produced by PCP reductions.
It is our hope that this result contributes to clarifying the landscape of quantum advantage for optimization problems of practical relevance~\cite{MindTheGaps,VastWorld,Abbas_2024,babbush2025grandchallengequantumapplications}.

\section*{Acknowledgements}
This work has been inspired by a recent Computer Science Stack Exchange post~\cite{csse_maxlinsat_dqi}, which has motivated us to work out the inapproximability of \maxlin{} in full.
The authors warmly thank Elies Gil-Fuster for helpful discussions and feedback on this work, and Noah Shutty for a string of highly helpful comments on an earlier version of this manuscript. Moreover, the authors acknowledge support by the BMFTR (PraktiQOM, QuSol, HYBRID++, PQ-CCA),
the Munich Quantum Valley, Berlin Quantum, the Quantum Flagship (Millenion, PasQuans2), the DFG (CRC 183, SPP 2514), the Clusters of Excellence (MATH+, ML4Q), and the European Research Council (DebuQC).

\bibliographystyle{quantum}
\bibliography{main}
\end{document}

%% file: commands.tex
\providecommand{\myvec}[1]{\ensuremath{\boldsymbol{#1}}}

\providecommand{\xx}{\ensuremath{\myvec{x}}}



%% file: main.bib
@article{Jordan2024DQI,
   title={Optimization by decoded quantum interferometry},
   volume={646},
   ISSN={1476-4687},
   url={http://doi.org/10.1038/s41586-025-09527-5},
   DOI={10.1038/s41586-025-09527-5},
   number={8086},
   journal={Nature},
   publisher={Springer Science and Business Media LLC},
   author={Jordan, Stephen P. and Shutty, Noah and Wootters, Mary and Zalcman, Adam and Schmidhuber, Alexander and King, Robbie and Isakov, Sergei V. and Khattar, Tanuj and Babbush, Ryan},
   year={2025},
   month=oct, pages={831–836} }

@article{khattar2025verifiablequantumadvantageoptimized,
      title={{Verifiable quantum advantage via optimized DQI circuits}}, 
      author={Tanuj Khattar and Noah Shutty and Craig Gidney and Adam Zalcman and Noureldin Yosri and Dmitri Maslov and Ryan Babbush and Stephen P. Jordan},
      year={2025},
      eprint={2510.10967},
      archivePrefix={arXiv},
      journal = {arXiv preprint arXiv:2510.10967},
      url={https://arxiv.org/abs/2510.10967}, 
}

@article{piveteau2025_quantum_decoding,
      title={Efficient and optimal quantum state discrimination via quantum belief propagation}, 
      author={Christophe Piveteau and Joseph M. Renes},
      year={2025},
      eprint={2509.19441},
      archivePrefix={arXiv},
      journal = {arXiv preprint arXiv:2509.19441},
      url={https://arxiv.org/abs/2509.19441}, 
}

@article{briaud2025quantumadvantagesolvingmultivariate,
      title={Quantum Advantage via Solving Multivariate Polynomials}, 
      author={Pierre Briaud and Itai Dinur and Riddhi Ghosal and Aayush Jain and Paul Lou and Amit Sahai},
      year={2025},
      eprint={2509.07276},
      archivePrefix={arXiv},
      journal = {arXiv preprint arXiv:2509.07276},
      url={https://arxiv.org/abs/2509.07276}, 
}

@article{sabater2025solvingindustrialintegerlinear,
      title={Towards solving industrial integer linear programs with Decoded Quantum Interferometry}, 
      author={Francesc Sabater and Ouns El Harzli and Geert-Jan Besjes and Marvin Erdmann and Johannes Klepsch and Jonas Hiltrop and Jean-Francois Bobier and Yudong Cao and Carlos A. Riofrio},
      year={2025},
      eprint={2509.08328},
      archivePrefix={arXiv},
      journal = {arXiv preprint arXiv:2509.08328},
      url={https://arxiv.org/abs/2509.08328}, 
}

@article{chailloux2024softdecoders,
      title={Quantum advantage from soft decoders}, 
      author={André Chailloux and Jean-Pierre Tillich},
      year={2024},
      eprint={2411.12553},
      archivePrefix={arXiv},
      journal = {arXiv preprint arXiv:2411.12553},
      url={https://arxiv.org/abs/2411.12553}, 
}

@inproceedings{Patamawisut2025,
   title={Quantum Circuit Design for Decoded Quantum Interferometry},
   url={http://doi.org/10.1109/QCE65121.2025.00041},
   DOI={10.1109/qce65121.2025.00041},
   booktitle={2025 IEEE Int.l Conf. Quant. Comp. Eng. (QCE)},
   publisher={IEEE},
   author={Patamawisut, Natchapol and Benchasattabuse, Naphan and Hajdušek, Michal and Van Meter, Rodney},
   year={2025},
   month=aug, pages={291–301} }

@article{bu2025DQInoise,
      title={Decoded Quantum Interferometry Under Noise}, 
      author={Kaifeng Bu and Weichen Gu and Dax Enshan Koh and Xiang Li},
      year={2025},
      eprint={2508.10725},
      archivePrefix={arXiv},
      journal = {arXiv preprint arXiv:2508.10725},
      url={https://arxiv.org/abs/2508.10725}, 
}

@article{ralli2025DQI,
author={Ralli, Alexis
and Weaving, Tim
and Coveney, Peter V.
and Love, Peter J.},
title={{Bridging quantum chemistry and MaxCut: Classical performance guarantees and quantum algorithms for the Hartree--Fock method}},
journal={J. Chem. Th. Comp.},
year={2025},
month={Oct},
day={14},
publisher={American Chemical Society},
volume={21},
number={19},
pages={9511-9524},
issn={1549-9618},
doi={10.1021/acs.jctc.5c00948},
url={https://doi.org/10.1021/acs.jctc.5c00948}
}

@article{berlekamp1978, title={On the inherent intractability of certain coding problems}, volume={24}, ISSN={0018-9448}, DOI={10.1109/TIT.1978.1055873}, abstractNote={The fact that the general decoding problem for linear codes and the general problem of finding the weights of a linear code are both NP-complete is shown. This strongly suggests, but does not rigorously imply, that no algorithm for either of these problems which runs in polynomial time exists.}, number={3}, journal={IEEE Trans. Inf. Th.}, publisher={IEEE}, author={Berlekamp, Elwyn R. and McEliece, Robert J. and van Tilborg, Henk C. A.}, year={1978}, month={May}, pages={384–386} }

@article{hastad2001, author = {H\r{a}stad, Johan}, title = {Some optimal inapproximability results}, year = {2001}, issue_date = {July 2001}, publisher = {Association for Computing Machinery}, address = {New York, NY, USA}, volume = {48}, number = {4}, issn = {0004-5411}, url = {https://doi.org/10.1145/502090.502098}, doi = {10.1145/502090.502098}, journal = {J. ACM}, month = jul, pages = {798–859}, numpages = {62} }

@misc{csse_maxlinsat_dqi,
  author={Kumar, Manish},
  title={{Hardness of approximation for Max-LinSAT mentioned in DQI paper}},
  year={2025},
  howpublished={\href{https://cs.stackexchange.com/questions/174832}{Computer Science Stack Exchange}},
  url={https://cs.stackexchange.com/questions/174832}
}

@article{marwaha2025complexitydecodedquantuminterferometry,
      title={On the Complexity of Decoded Quantum Interferometry}, 
      author={Kunal Marwaha and Bill Fefferman and Alexandru Gheorghiu and Vojtech Havlicek},
      year={2025},
      eprint={2509.14443},
      archivePrefix={arXiv},
      journal = {arXiv preprint arXiv:2509.14443},
      url={https://arxiv.org/abs/2509.14443}, 
}

@article{Boixo,
title={Characterizing quantum supremacy in near-term devices},
doi={10.1038/s41567-018-0124-x},
author={Sergio Boixo and Sergei V. Isakov and Vadim N. Smelyanskiy and Ryan Babbush and Nan Ding and Zhang Jiang and Michael J. Bremner and John M. Martinis and Hartmut Neven},
journal={Nature Phys.},volume= 14, pages={595-600}, year= 2018}

@article{SupremacyReview,
  title = {Computational advantage of quantum random sampling},
  author = {Hangleiter, Dominik and Eisert, Jens},
  journal = {Rev. Mod. Phys.},
  volume = {95},
  pages = {035001},
  year = {2023},
  publisher = {American Physical Society},
  doi = {10.1103/RevModPhys.95.035001}
}

@article{anschuetz2025DQI,
      title={Decoded Quantum Interferometry Requires Structure}, 
      author={Eric R. Anschuetz and David Gamarnik and Jonathan Z. Lu},
      year={2025},
      eprint={2509.14509},
      archivePrefix={arXiv},
      journal = {arXiv preprint arXiv:2509.14509},
      url={https://arxiv.org/abs/2509.14509}, 
}

@article{parekh2025DQI_maxcut,
      title={No Quantum Advantage in Decoded Quantum Interferometry for MaxCut}, 
      author={Ojas Parekh},
      year={2025},
      eprint={2509.19966},
      archivePrefix={arXiv},
      journal = {arXiv preprint arXiv:2509.19966},
      url={https://arxiv.org/abs/2509.19966}, 
}

@article{gu2025algebraicgeometrycodesdecoded,
      title={Algebraic Geometry Codes and Decoded Quantum Interferometry}, 
      author={Andi Gu and Stephen P. Jordan},
      year={2025},
      eprint={2510.06603},
      journal = {arXiv preprint arXiv:2510.06603},
      url={https://arxiv.org/abs/2510.06603}, 
}

@article{bu2026hamiltoniandecodedquantuminterferometry,
      title={{Hamiltonian decoded quantum interferometry for general Pauli Hamiltonians}}, 
      author={Kaifeng Bu and Weichen Gu and Xiang Li},
      year={2026},
      eprint={2601.18773},
      archivePrefix={arXiv},
      journal = {arXiv preprint arXiv:2601.18773},
      url={https://arxiv.org/abs/2601.18773}, 
}

@article{schmidhuber2025hamiltoniandecodedquantuminterferometry,
      title={Hamiltonian Decoded Quantum Interferometry}, 
      author={Alexander Schmidhuber and Jonathan Z. Lu and Noah Shutty and Stephen Jordan and Alexander Poremba and Yihui Quek},
      year={2025},
      eprint={2510.07913},
      archivePrefix={arXiv},
      journal = {arXiv preprint arXiv:2510.07913},
      url={https://arxiv.org/abs/2510.07913}, 
}

@article{kothari2025exponentialquantumspeedupmathrmsisinfty,
      title={No exponential quantum speedup for $\mathrm{SIS}^\infty$ anymore}, 
      author={Robin Kothari and Ryan O'Donnell and Kewen Wu},
      year={2025},
      eprint={2510.07515},
      archivePrefix={arXiv},
      journal = {arXiv preprint arXiv:2510.07515},
      url={https://arxiv.org/abs/2510.07515}, 
}

@article{rosmanis2026nearlylineartimedecodedquantum,
      title={A nearly linear-time Decoded Quantum Interferometry algorithm for the Optimal Polynomial Intersection problem}, 
      author={Ansis Rosmanis},
      year={2026},
      eprint={2601.15171},
      archivePrefix={arXiv},
      journal = {arXiv preprint arXiv:2601.15171},
      url={https://arxiv.org/abs/2601.15171}, 
}

@article{Prange,
author={Eugene Prange},
doi={10.1109/TIT.1962.1057777},
title={The use of information sets in decoding cyclic codes},
journal={IRE Trans. Inf. Th.}, volume=8, pages={S5–S9}, year=1962, nolink = {}}

@article{chailloux2025opixsoftdecoders,
      title={{OPI x soft decoders}}, 
      author={André Chailloux},
      year={2025},
      eprint={2511.22691},
      archivePrefix={arXiv},
      journal = {arXiv preprint arXiv:2511.22691},
      url={https://arxiv.org/abs/2511.22691}, 
}

@inproceedings{Raghavendra2005,
author = {Raghavendra, Prasad},
title = {Optimal algorithms and inapproximability results for every CSP?},
year = {2008},
isbn = {9781605580470},
publisher = {Association for Computing Machinery},
address = {New York, NY, USA},
url = {https://doi.org/10.1145/1374376.1374414},
doi = {10.1145/1374376.1374414},
booktitle = {Proceedings of the Fortieth Annual ACM Symposium on Theory of Computing},
pages = {245–254},
numpages = {10},
location = {Victoria, British Columbia, Canada},
series = {STOC '08}
}

@inproceedings{Austrin2008,
author = {Austrin, Per and Mossel, Elchanan},
title = {Approximation Resistant Predicates from Pairwise Independence},
year = {2008},
isbn = {9780769531694},
publisher = {IEEE Computer Society},
address = {USA},
url = {https://doi.org/10.1109/CCC.2008.20},
doi = {10.1109/CCC.2008.20},
booktitle = {Proceedings of the 2008 IEEE 23rd Annual Conference on Computational Complexity},
pages = {249–258},
numpages = {10},
series = {CCC '08}
}

@article{Chan2016,
author = {Chan, Siu On},
title = {Approximation Resistance from Pairwise-Independent Subgroups},
year = {2016},
issue_date = {September 2016},
publisher = {Association for Computing Machinery},
address = {New York, NY, USA},
volume = {63},
number = {3},
issn = {0004-5411},
url = {https://doi.org/10.1145/2873054},
doi = {10.1145/2873054},
journal = {J. ACM},
month = aug,
articleno = {27},
numpages = {32}
}

@article{Engebretsen2004,
title = {Inapproximability results for equations over finite groups},
journal = {Theoretical Computer Science},
volume = {312},
number = {1},
pages = {17-45},
year = {2004},
note = {Automata, Languages and Programming},
issn = {0304-3975},
doi = {https://doi.org/10.1016/S0304-3975(03)00401-8},
url = {https://www.sciencedirect.com/science/article/pii/S0304397503004018},
author = {Lars Engebretsen and Jonas Holmerin and Alexander Russell}
}

@InProceedings{Butti2025,
  author =	{Butti, Silvia and Larrauri, Alberto and \v{Z}ivn\'{y}, Stanislav},
  title =	{{Optimal inapproximability of promise equations over finite groups}},
  booktitle =	{52nd International Colloquium on Automata, Languages, and Programming (ICALP 2025)},
  pages =	{38:1--38:14},
  series =	{Leibniz International Proceedings in Informatics (LIPIcs)},
  ISBN =	{978-3-95977-372-0},
  ISSN =	{1868-8969},
  year =	{2025},
  volume =	{334},
  editor =	{Censor-Hillel, Keren and Grandoni, Fabrizio and Ouaknine, Jo\"{e}l and Puppis, Gabriele},
  publisher =	{Schloss Dagstuhl -- Leibniz-Zentrum f{\"u}r Informatik},
  address =	{Dagstuhl, Germany},
  url =		{https://drops.dagstuhl.de/entities/document/10.4230/LIPIcs.ICALP.2025.38},
  URN =		{urn:nbn:de:0030-drops-234150},
  doi =		{10.4230/LIPIcs.ICALP.2025.38}
}

@article{AS98,
author = {Arora, Sanjeev and Safra, Shmuel},
title = {{Probabilistic checking of proofs: a new characterization of NP}},
year = {1998},
issue_date = {Jan. 1998},
publisher = {Association for Computing Machinery},
address = {New York, NY, USA},
volume = {45},
number = {1},
issn = {0004-5411},
url = {https://doi.org/10.1145/273865.273901},
doi = {10.1145/273865.273901},
journal = {J. ACM},
month = jan,
pages = {70–122},
numpages = {53}
}

@article{ALMSS98,
author = {Arora, Sanjeev and Lund, Carsten and Motwani, Rajeev and Sudan, Madhu and Szegedy, Mario},
title = {Proof verification and the hardness of approximation problems},
year = {1998},
issue_date = {May 1998},
publisher = {Association for Computing Machinery},
address = {New York, NY, USA},
volume = {45},
number = {3},
issn = {0004-5411},
url = {https://doi.org/10.1145/278298.278306},
doi = {10.1145/278298.278306},
journal = {J. ACM},
month = may,
pages = {501–555},
numpages = {55}
}

@article{
Pirnay2024,
author = {Niklas Pirnay  and Vincent Ulitzsch  and Frederik Wilde  and Jens Eisert  and Jean-Pierre Seifert },
title = {An in-principle super-polynomial quantum advantage for approximating combinatorial optimization problems via computational learning theory},
journal = {Science Adv.},
volume = {10},
number = {11},
pages = {eadj5170},
year = {2024},
doi = {10.1126/sciadv.adj5170},
url= {https://www.science.org/doi/abs/10.1126/sciadv.adj5170}}

@article{MindTheGaps,
      title={Mind the gaps: The fraught road to quantum advantage},
    author={Jens Eisert and John Preskill}, 
      year={2025},
      journal = {arXiv preprint arXiv:2510.19928},
      eprint={2510.19928},
      archivePrefix={arXiv}
}

@article{VastWorld,
  archiveprefix = {arXiv},
  eprint = {2508.05720}, 
  year=2025,
    title={The vast world of quantum advantage},
    journal = {arXiv preprint arXiv:2508.05720},
    Author={Hsin-Yuan Huang and Soonwon Choi and Jarrod R. McClean and John Preskill}}

@article{Yamakawa2024,
   title={Verifiable quantum advantage without structure},
   volume={71},
   ISSN={1557-735X},
   url={http://dx.doi.org/10.1145/3658665},
   DOI={10.1145/3658665},
   number={3},
   journal={J. ACM},
   publisher={Association for Computing Machinery (ACM)},
   author={Yamakawa, Takashi and Zhandry, Mark},
   year={2024},
   month=jun, pages={1–50} }

@article{szegedy2022,
  title         = {Quantum advantage for combinatorial optimization problems, Simplified},
  author        = {Szegedy, Mario},
  year          = {2022},
  journal       = {arXiv preprint arXiv:2212.12572},
  archivePrefix = {arXiv},
  eprint        = {2212.12572},
  url           = {https://arxiv.org/abs/2212.12572},
}

@article{buhrman2025formalframeworkquantumadvantage,
  title         = {Formal Framework for Quantum Advantage},
  author        = {Buhrman, Harry and Galke, Niklas and Meichanetzidis, Konstantinos},
  year          = {2025},
  journal       = {arXiv preprint arXiv:2510.01953},
  archivePrefix = {arXiv},
  eprint        = {2510.01953},
  url           = {https://arxiv.org/abs/2510.01953},
}

@article{Abbas_2024,
   title={Challenges and opportunities in quantum optimization},
   volume={6},
   ISSN={2522-5820},
   url={http://dx.doi.org/10.1038/s42254-024-00770-9},
   DOI={10.1038/s42254-024-00770-9},
   number={12},
   journal={Nature Rev. Phys.},
   publisher={Springer Science and Business Media LLC},
   author={Abbas, Amira and Ambainis, Andris and Augustino, Brandon and Bärtschi, Andreas and Buhrman, Harry and Coffrin, Carleton and Cortiana, Giorgio and Dunjko, Vedran and Egger, Daniel J. and Elmegreen, Bruce G. and Franco, Nicola and Fratini, Filippo and Fuller, Bryce and Gacon, Julien and Gonciulea, Constantin and Gribling, Sander and Gupta, Swati and Hadfield, Stuart and Heese, Raoul and Kircher, Gerhard and Kleinert, Thomas and Koch, Thorsten and Korpas, Georgios and Lenk, Steve and Marecek, Jakub and Markov, Vanio and Mazzola, Guglielmo and Mensa, Stefano and Mohseni, Naeimeh and Nannicini, Giacomo and O’Meara, Corey and Tapia, Elena Peña and Pokutta, Sebastian and Proissl, Manuel and Rebentrost, Patrick and Sahin, Emre and Symons, Benjamin C. B. and Tornow, Sabine and Valls, Víctor and Woerner, Stefan and Wolf-Bauwens, Mira L. and Yard, Jon and Yarkoni, Sheir and Zechiel, Dirk and Zhuk, Sergiy and Zoufal, Christa},
   year={2024},
   month=oct, pages={718–735} }

@inproceedings{Khot2002UGC, author = {Khot, Subhash}, title = {On the power of unique 2-prover 1-round games}, year = {2002}, isbn = {1581134959}, publisher = {Association for Computing Machinery}, address = {New York, NY, USA}, url = {https://doi.org/10.1145/509907.510017}, doi = {10.1145/509907.510017}, booktitle = {Proceedings of the Thiry-Fourth Annual ACM Symposium on Theory of Computing}, pages = {767–775}, numpages = {9}, location = {Montreal, Quebec, Canada}, series = {STOC '02} }

@article{babbush2025grandchallengequantumapplications,
      title={The Grand Challenge of Quantum Applications}, 
      author={Ryan Babbush and Robbie King and Sergio Boixo and William Huggins and Tanuj Khattar and Guang Hao Low and Jarrod R. McClean and Thomas O'Brien and Nicholas C. Rubin},
      year={2025},
      eprint={2511.09124},
      archivePrefix={arXiv},
      journal = {arXiv preprint arXiv:2511.09124},
      url={https://arxiv.org/abs/2511.09124}, 
}

@article{RS_nphard,
  author={Guruswami, Venkatesan and Vardy, Alexander},
  journal={IEEE Trans. Inf. Th.}, 
  title={{Maximum-likelihood decoding of Reed-Solomon codes is NP-hard}}, 
  year={2005},
  volume={51},
  number={7},
  pages={2249-2256},
  doi={10.1109/TIT.2005.850102}}
